\documentclass[journal]{IEEEtran}
\makeatletter
\usepackage{setspace}
\usepackage[final]{graphicx}
\usepackage[reqno]{amsmath}
\usepackage{amssymb}
\usepackage{cite}
\usepackage{balance}
\usepackage{soul, color}
\usepackage[inoutnumbered,linesnumbered,algoruled,slide,vlined]{algorithm2e}
\usepackage{algorithmic}
\usepackage{subfig}
\usepackage{floatrow}
\usepackage{lipsum}
\newfont{\bbb}{msbm10 scaled 500}

\newfont{\bb}{msbm10 scaled 1100}

\newcommand{\EE}{\mbox{\bb E}}

\newcommand{\Rc}{{\cal R}}
\newcommand{\Bc}{{\cal B}}

\newcommand{\argmax}{\operatornamewithlimits{argmax}}

\newtheorem{theorem}{Theorem}
\newtheorem{lemma}{Lemma}%[chapter]
\begin{document}

%\title{ Achieving the Centralized Performance with CSMA based Scheduling for Wireless Local Area Networks}
%\author{\IEEEauthorblockN{Mehmet Karaca and Bj{\"o}rn Landfeldt \\}
%\IEEEauthorblockA{The Department of Electrical and Information Technologies, Lund University, Lund/Sweden. \\Email: mehmet.karaca@eit.lth.se, bjorn.landfeldt@eit.lth.se}\thanks{This  work  was  partially  sponsored  by  the  EC  FP7  Marie Curie IAPP Project 324515, ”MeshWise"}}

\title{  Approaching  Optimal Centralized Scheduling with CSMA-based  Random Access   over Fading Channels
    \thanks{
    Mehmet Karaca and Bj{\"o}rn Landfeldt are with the Department  of Electrical and Information Technology, Lund University, Lund, Sweden. Email: \{mehmet.karaca, bjorn.landfeldt\}@eit.lth.se.}
 \thanks{This work was sponsored by the EC  FP7  Marie Curie IAPP Project 324515, ``MeshWise".}}
\author{\IEEEauthorblockN{Mehmet Karaca and Bj{\"o}rn Landfeldt \\}
\IEEEauthorblockA{}}

\maketitle
\begin{abstract}
Carrier Sense Multiple Access (CSMA) based distributed  algorithms can attain the largest capacity region as the centralized Max-Weight  policy does. Despite their capability of achieving throughput-optimality, these algorithms can either incur large delay and have large complexity or only operate over non-fading channels. In this letter, by assuming arbitrary back-off time we first propose a fully distributed randomized algorithm whose performance can be pushed to the performance of the centralized Max-Weight policy  not only in terms of throughput  but also in terms of delay for completely-connected interference networks with fading channels. Then, inspired by the proposed algorithm we introduce an implementable distributed algorithm for practical networks with a reservation scheme. We show that the proposed practical algorithm can still achieve the performance of the centralized Max-Weight policy. 
\end{abstract}
\vspace{-0.3cm}
\section{Introduction}
The major challenge of  channel access management in wireless networks is optimally scheduling users so that interference is eliminated and the maximum performance in terms of throughput, delay, jitter, etc., can be achieved. The centralized Max-Weight algorithm (MW)~\cite{MW} can achieve maximum throughput, and can stabilize the network. However, due to the huge complexity of MW algorithm Carrier
Sense Multiple Access (CSMA) based algorithms have received significant attention from many researchers. In~\cite{Jiang:Allerton08},~\cite{Ni:Infocom10},~\cite{Rajagopalan:Sigmetric09} and~\cite{Shah:focs11}, CSMA based throughput optimal algorithms are developed for nonfading channels, which, however, suffer from either poor delay performance or high complexity. The authors in~\cite{Li:TWC} present an optimal scheduling algorithm for deadline constrained-traffic by assuming continuous back-off time  which is impossible to implement, and also the optimality does not hold for practical systems. Also, the performance of the discrete time version of the algorithm in~\cite{Li:TWC} can be quite poor. The authors in~\cite{Huang:Net} try to address the poor delay performance of CSMA based algorithms by assuming non-fading channels. The authors in~\cite{Karaca:gc15} and~\cite{Chiang:secon13} try to develop MW type algorithms for practical 802.11 networks without providing any analytical guarantees. 

In order to be feasible and practical for wireless systems, any CSMA based algorithm should; i-) have low-complexity, and require only local information (no message passing); ii-) have good delay performance; iii-)  perform well over general channel and arrival conditions. Since all these conditions can be satisfied by a centralized system with complete network information, we ask the following question: \textit{Is it possible that a CSMA based distributed algorithm  can achieve the same performance as the centralized solution in all these dimensions under certain conditions}?  We find that the answer to this question is yes under the condition of a complete graph network. Our contributions are summarized as follows; i) when arbitrary back-off is allowed, we first develop an optimal distributed algorithm which schedules the user with the maximum weighted
rate at every time slot without requiring any message passing over general fading and traffic conditions as the centralized MW algorithm does.  Hence, the same performance as that of the centralized MW policy is indeed achievable; ii) Then,  we design an implementable distributed
algorithm to be compatible with practical systems, which can still schedule the user with the maximum weighted rate with low overhead associated with  contention resolution.
\vspace{-0.5cm}
\section{System Model}
We consider a  network where $N$ users contend for transmission in the same contention domain. Time
is slotted, $t\in \{0,1,2,\ldots\}$, and each user's wireless channel is assumed to be independent across
users and time. The gain of the channel is constant over the
duration of a time slot but varies between slots. As in practice, we consider that only a fixed set of data rates
$\Rc=\{r_1,r_2,\ldots,r_L\}$ can be supported. The
distribution of the channel rates for user $n$ is denoted as
$\Pr[R_n(t)=r_l]=p_{n}^l$, for all $l=\{1,\ldots,L\}$, and $R_n(t) \in \Rc$ for all $n$ and $t$. Let
 $I_n(t)$ be the scheduler decision, where
$I_n(t)=1$ if user  $n$ is scheduled for transmission in slot $t$,
and $I_n(t)=0$ otherwise. Each user maintains a separate queue, and packets arrive
according a stationary arrival process that is independent across
users and time slots. Let $A_n(t)$ be the amount of data arriving
into the queue of user $n$ at time slot $t$. Let $Q_n(t)$ and
$R_n(t)$ denote the queue length and transmission rate of user $n$
at time $t$, respectively. The queue
length variation for user $n$ is given as $Q_n(t+1)=[Q_n(t)+A_n(t)-R_{n}(t)I_n(t)]^+$,
%\begin{align}
%Q_n(t+1)=[Q_n(t)+A_n(t)-R_{n}(t)I_n(t)]^+ , \label{eq:queuelength}
%\end{align}
where $[y]^+=\max(y,0)$. Let $w_n(t)$ be the weighted rate of user $n$ at time slot $t$,
where $w_n(t)=Q_n(t)R_n(t)$. In their seminal paper, Tassiulas and
Ephremides~\cite{MW} have shown that Max-Weight algorithm scheduling
the user $k=\argmax_n w_n(t)$ at every time slot can stabilize  the network
whenever this is possible, and we denote $k$ as the user with the highest weighted rate at each slot. Unlike previous work, our aim is to develop a fully distributed algorithm that guarantees to schedule user $k$ at every time slot. Next, we give our Distributed MW policy with Arbitrary Back-off Time (DMW-AB).
\vspace{-0.3cm}
\section{DMW with Arbitrary Backoff}
The basic principle behind our algorithm is to determine a back-off procedure
in CSMA/CA (CSMA/Collision Avoidance) systems where users' back-off
time is determined by their queue sizes and channel
conditions. 
\textbf{DMW-AB} is performed in a two step
procedure: In the first step, each user $n \in \{1,2,\cdots,N\}$
generates an exponentially distributed random variable denoted by
$X_n(t)$ with rate $r_n(t)=f(w_n(t))$, which gives the random back-off
duration of user $n$ at that time slot. Then, if user $n$ does not
sense any other transmission until its random back-off duration
expires, it starts transmitting.  Otherwise it keeps silent.

Clearly, the performance of DMW-AB depends on the function $f(x)$. Next, we give our main theorem which will characterize $f(x)$ to make DMW-AB throughout-optimal. We note that our $f(x)$ function may not be unique, and there may be other functions that satisfy our theoretical results.
\begin{theorem}
\label{thm:DMWAB} Given any $b>1$, if $f(x)=b^x$ and
$R_n(t)\geq R_{min}$, and $R_{min} >0$, $\forall n, t$ then DMW-AB is throughput-optimal. 
\end{theorem}
\begin{proof}
The proof is provided in Appendix \ref{sec:DMWAB}.
\end{proof}
Although for any $b>1$, DMW-AB is optimal, the delay performance of it can be poor since Theorem 1 only holds when the queue sizes are sufficiently large.  However, in Lemma 1 we will show that the condition on the queue sizes can become less restrictive for large values of $b$ by guaranteeing to schedule the user $k$  at every time slot.
\begin{lemma}
\label{lemma:1} When $f(x)=b^x$ and as $b \rightarrow \infty$, the
probability that the user with the maximum weighted rate is
scheduled by DWM-AB algorithm at each time slot goes to 1.
\end{lemma}
\begin{proof}
The probability that the back-off timer of user $k$ expires first under DWM-AB algorithm is 
given by $\pi_{k}(t)=\frac{b^{w_{k}(t)}}{\sum_{n=1}^N b^{w_n(t)}}$
%\begin{align*}
%\pi_{k}(t)=\frac{b^{w_{k}(t)}}{\sum_{n=1}^N b^{w_n(t)}}
%\end{align*}
By dividing both numerator and denominator by $b^{w_{k}(t)}$, we
have,
\begin{align*}
\pi_{k}(t)=\frac{1}{1 + \sum_{\substack{\scalebox{0.7}{n=1}\\ \scalebox{0.7}n \neq \scalebox{0.7} k}}^N b^{w_n(t)-w_{k}(t)}}
\end{align*}
We know that $w_{k}(t) > w_n(t)$ $\forall n \neq k$. Hence,
taking $b \rightarrow \infty$ yields that $\lim_{b \rightarrow\infty}  \pi_{k}=\frac{1}{1 +0}=1$.
%\begin{align*}
%\lim_{b \rightarrow\infty}  \pi_{k}=\frac{1}{1 +0}=1
%\end{align*}
Thus, DMW-AB guarantees to schedule the user with the maximum weight. This completes the proof.
%The proof is given in Appendix \ref{sec:lemma1}.
\end{proof}
With Theorem 1 and Lemma 1, we show that DMW-AB can achieve the same performance as that of the centralized MW policy in terms of the average delay and throughput.  Recall that DMW-AB algorithm assumes that users's back-off time can
take any positive values. However, this assumption does not hold 
in practical systems. For instance, in current IEEE 802.11
networks a discrete-time backoff scale is used. Next,
we introduce the practical version of DMW-AB.
\vspace{-0.32cm}
\section {DMW with Reservation }
In order to design a practical implementable algorithm, we
divide each time slot into a control slot and a data slot  whose duration depends on the transmission duration at that time. The
control slot consists of mini-slots that enable a collision-free
data transmission.  
Let $M(t)$ be the number of mini-slots used for
contention resolution at time $t$, which $M(t)$ is a random variable
depending on the weighted throughput of users, and let  $\EE[M(t)]$ be the expected number of
$M(t)$. Our practical algorithm namely DMW with Reservation Scheme (DMW-RS) is based on a reservation form,  in which at the beginning of
each mini-slot each user has to decide whether or not to attempt
transmission in the current data slot. Let $a_n(m, t)$ be the
attempt decision of user $n$ at mini-slot $m$ and time slot $t$. If
user $n$ attempts to transmit then $a_n(m,t)=1$. Otherwise,
$a_n(m,t)=0$. Let $\tau(t)$ be a network parameter which
determines the attempt decision of users as we will explain next.

The \textbf{DMW-RS} algorithm is performed as follows: First, at
the beginning of each mini-slot $m$, each user $n \in \{1,2,
\cdots,N\}$ generates an exponential random variable $X_n(m,t)$ with
rate $r_n(t)=b^{w_n(t)}$.  If $X_n(m,t) < \tau(t)$, then user  $n$
announces its intent by sending a short message (e.g., Request-to-Send Messages (RTS)), and
$a_n(m,t)=1$. Otherwise, user $n$ stays idle (i.e., $a_n(m,t)=0$)
and skips to the next mini-slot. If more than one user 
attempts transmission at mini-slot $m$, there will be a collision. Also, if any user does not intend to transmit then that mini-slot will be idle. In case of a collision or idle mini slot,  each user $n$ generates  a new exponential random
variable with rate $r_n(t)=b^{w_n(t)}$, follows the same procedure
and tries again. This process is iterated until there is only one
user attempting to access the channel, and that user transmits its
data in the data slot. %Note that $r_n(t)$ does not depend on index $m$ since we
%assume that the duration of mini-slot is short (i.e., IEEE 802.11b
%20$\mu$s) such that queue size and transmission rate of users do not
%change.
%Different from DMW-AB,  DMW-RS introduces some amount of
%overhead in terms of mini-slots for contention resolution. Now, we turn our attention to
%develop an efficient method so that $\EE[M(t)]$ is reduced as much
%as possible. 
We next show that DMW-RS can schedule the user with the maximum weighted throughout at every time.
The probability that only the user $k$ intents (e.g., only the user $k$ sends RTS packet) and
other users keep silent at a mini-slot is given by,
\begin{align}
P_{k}(t)=\left(1-e^{-\tau(t) r_k(t) }\right)\prod_{\substack{n=1\\ n \neq k}}^N
e^{-\tau(t) r_n(t)}\label{eq:success}.
\end{align}
One can show that  $P_{k}(t)$ is a concave function of $\tau(t)$.
Let $\tau^*(t)$ denote the value taken by $\tau(t)$ at the maximum
value of $P_{k}(t)$. When all queue sizes, channel rates and $b$ are given,
$\tau^*(t)$ is found by taking the first derivative of $P_{k}(t)$,
setting it to zero and solving for $\tau(t)$. Then, we have 
\begin{align}
\tau^*(t)=\frac{1}{r_{k}(t)}\ln\left( 1 + \frac{r_{k}(t)} {\sum_{\substack{\scalebox{0.7}{n=1}\\ \scalebox{0.7}n \neq \scalebox{0.7} k}}^N
r_n(t)} \right) \label{eq:tau}.
\end{align}
When we plug $\tau^*(t)$ into equation \eqref{eq:success}, we obtain
the maximum value of $P_k(t)$ denoted by $P_k^*(t)$ as follows:
\begin{align}
P_{k}^*(t)=\left( \frac{\sum_{\substack{\scalebox{0.7}{n=1}\\ \scalebox{0.7}n \neq \scalebox{0.7} k}}^N r_n(t)}  {\sum_{n=1}^N
r_n(t)} \right)^{\frac{\sum_{n \neq k}^N r_n(t)} {r_{k}}}\left(
\frac{r_k(t)}{\sum_{n=1}^N r_n(t)} \right)\label{eq:success2}.
\end{align}
The  limit $\lim_{b \rightarrow \infty} P_{k}^*(t)$ can be directly evaluated by using
L'Hospital's rule,
%\begin{align}
%\lim_{b \rightarrow \infty} P_{k}^*(t) \label{eq:limit}.
%\end{align}
%The limit in \eqref{eq:limit} can be directly evaluated by using
%L'Hospital's rule, 
and one can show that as $b \rightarrow
\infty$ then $P_k^*(t) \rightarrow 1$. Hence, the user with the
maximum weighted throughput is guaranteed to be scheduled with DMW-RS. 

Different from DMW-AB,  DMW-RS introduces some amount of
overhead in terms of mini-slots for contention resolution. Now, we turn our attention to
develop an efficient method so that $\EE[M(t)]$ is minimized as much
as possible. We note that to minimize $\EE[M(t)]$ we need to find $\tau^*(t)$. Clearly,
in order to find the optimal $\tau^*(t)$, global knowledge of
the network is required, i.e., $w_n(t)$ for all $n$. Next, our goal is to find a method for tracking and estimating
$\tau^*(t)$ without requiring global network information.

In practice, users have finite buffer size. Let $B$\footnote{$B$ should be sufficiently large to avoid any throughput loss. In our algorithms since as $b$ increases the condition on the maximum queue size in Theorem 1 becomes less restrictive, then the same performance can be achieved even with small buffer sizes.}  be the maximum buffer size of each user. Then, the maximum weighted rate that a user can achieve at any
time slot is equal to $w_{m}=B \times r_L$ since at most $B$
packets can be kept at each queue buffer and the maximum
transmission rate is $r_L$. Let $\tau^*_l $ and $\tau^*_u$ be the
minimum and maximum values that $\tau^*(t)$ can take at any time
slot $t$, respectively. From \eqref{eq:tau}, we approximate that $\tau^*_l$ is achieved nearly
when $w_n(t)=w_{m}$ for all $n$. Thus,
$\tau^*_l=\frac{1}{b^{w_{m}}}\ln\left( 1 + \frac{1} {N-1} \right)$.
Also, $\tau^*_u$ occurs  when $w_n(t)=0$ for all $n$. We note that this approximation is more accurate when $N$ is sufficiency large. Then,
$\tau^*_u=\ln\left( 1 + \frac{1} {N-1} \right)$. The range of
$\tau^*(t)$ is given as $\tau^*_l \leq \tau^*(t)\leq \tau^*_u  \quad \forall t$.
%\begin{align}
%\tau^*_l \leq \tau^*(t)\leq \tau^*_u  \quad \forall t.
%\label{eq:taurange}
%\end{align}

Let $\tilde{\tau}(m,t)$ be
the estimated\footnote{The estimation requires the knowledge of $N$, which can be done by each user online~\cite{Bianchi:Nestimate}. In this work, we assume that $N$ is fixed, and known by each user. } $\tau^*(t)$ at mini-slot $m$ and time slot $t$, and let
$\delta$ be the constant estimation parameter, and $\Bc=\{b_1, b_2, \dots, b_V \}$ be the set of $b$ values, where $b_1>1$, and $b_g>b_h$, if $g>h$.  The estimation follows two steps:   first, we estimate $\tau^*(t)$  by decreasing  $\tilde{\tau}(m,t)$  when a collision occurs at mini-slot $m$ and
increasing it if no user transmits. In the second step, if the number of consecutive collisions is higher than a threshold denoted by $collthr$, then decrease $b$ to further reduce contention. If the number of consecutive idle slots is higher than a threshold denoted by $idlethr$ then increase $b$ so that the channel attempt probability increases. In Algorithm 1, the DMW-RS algorithm with updating $\tau$ is given.
% \begin{algorithm}
%\begin{itemize}
%\item {\color{red}\st{ Step 1: At the beginning of each time slot $t$, every user sets $\tilde{\tau}(1,t)=\tau^*_u$.}}
%\item  {\color{red}\st{ Step 2: Then, at each mini slot $m$, every user generates an exponential random variable $X_n(m,t)$ with
%rate, $r_n(t)=b^{w_n(t)}$.}}
%\item {\color{red}\st{ Step 3:  If $X_n(m,t) < \tilde{\tau}(m,t)$, $a_n(m,t)=1$. Otherwise
%$a_n(m,t)=0$.}}
%\begin{itemize}
%\item  {\color{red}\st{If $\sum_{n=1}^N a_n(m,t)>1$, there
%will be a collision at that mini slot. Then, update $\tilde{\tau}(m,t)
%\rightarrow \frac{\tilde{\tau}(m,t)}{g(\alpha(m,t))}$ where $g(\alpha(m,t))> 1$, and $m \rightarrow m+1$. Go Step 2.}}
%\item {\color{red}\st{If $\sum_{n=1}^N a_n(m,t)=0$, update $\tilde{\tau}(m,t) \rightarrow g(\alpha(m,t))\tilde{\tau}(m,t)$, and $m \rightarrow m+1$. Go Step 2.}}
%\end{itemize}
%\item {\color{red}\st{Step 4: If $\sum_{n=1}^N a_n(m,t)=1$, then the  user that attempts transmits at time slot $t$.}}
%\end{itemize}
%\caption{\small{{\color{red}\st{DMW-RS with $\tau$ Update for user $n$ at time $t$}}}}
%\label{tauupdate}
%\end{algorithm}
\begin{algorithm}
\begin{itemize}
\item Step 0:  At $t=0$, every user sets $v=V$, $b=b_V$ $\tilde{\tau}(1,t)=\frac{1}{b^{\alpha}}\ln\left( 1 + \frac{1} {N-1} \right)$, $ \alpha=w_m$, 
\item  Step 1: Then, at each mini slot $m$, user $n$ generates an exponential random variable $X_n(m,t)$ with
rate, $r_n(t)=b^{w_n(t)}$.
\item Step 2:  If $X_n(m,t) < \tilde{\tau}(m,t)$, $a_n(m,t)=1$. Otherwise
$a_n(m,t)=0$.
\begin{itemize}
\item  If $\sum_{n=1}^N a_n(m,t)>1$, collision occurs. Then, update: $\alpha \gets \alpha + \delta$,  $\tilde{\tau}(m,t)
\gets \frac{1}{b^{\alpha}}\ln\left( 1 + \frac{1} {N-1} \right)$, $m \gets m+1$. If $NumConsColl > collthr$,  then $b  \gets b_{v-1}$. Go Step 1.
\item If $\sum_{n=1}^N a_n(m,t)=0$,  idle slot and update: $\alpha \gets \alpha - \delta$.  $\tilde{\tau}(m,t)
\gets \frac{1}{b^{\alpha}}\ln\left( 1 + \frac{1} {N-1} \right)$, $m \gets m+1$. If $NumIdleColl > idlethr$, then $b  \gets b_{v+1}$ and Go Step 1.
\end{itemize}
\item Step 3: If $ a_n(m,t)=1$ and $a_k(m,t)=0$ for all $k \neq n$, then user $n$ transmits. Let $\tau^*_{prev}$ be the $\tau$ value at which  the contention is resolved at time $t$. Set $t \gets t+1$,   $\tilde{\tau}(1,t+1) \gets \tau^*_{prev}$, and $b \gets b_V$. Update the weight and Go Step 1.
\end{itemize}
\caption{\small{DMW-RS with $\tau$ Update for user $n$ at time $t$}}
\label{tauupdate}
\end{algorithm}
We note that the algorithm keeps tracking the number of consecutive collisions and idle slots, and checks if the boundary conditions for $\tau$ are satisfied.  Also, in Step 3,  at the beginning of next data slot, the algorithm uses the $\tau$ value at which the contention is resolved in the previous data slot. This is because if the network dynamics change slowly or never change (e.g., channel
states and queue process change slowly from time slot $t$ to time
slot $t+1$), then one can expect that $\tau^*(t+1)$ will be similar to the value
of $\tau^*(t)$ with high probability. 
\vspace{-0.5cm}
\section{Numerical Results}
\label{sec:sim} In this section, we present our MATLAB simulation results. There are 20 users transmitting to an Access Point. We set $B=200$ packets. Packets
arrive independently at each slot according to a Poisson distribution for each user. The channel changes independently over users and time, we set  $\Rc=\{1,2,3,4,5\}$ (i.e., $L=5$) in packets.  We consider that users have heterogeneous channels, and we divide the users into two groups where there are ten users in each
group. The channel state distributions of each group are given as
follows: for the first group, $n=\{1,2,\ldots,10\}$,  $l=\{1,\ldots,5\}$, $p_n^l=\{0.15, 0.2, 0.2, 0.15, 0.3\}$,  for the second group,  $n=\{11,\ldots,20\}$,  $l=\{1,\ldots,5\}$, $p_n^l=\{0.25, 0.25, 0.15,
0.1, 0.25\}$. We set  $\delta$=2,  $collthr=idlethr=7$  for all $m$ and $t$ and for DMW-RS  $\Bc=\{1.1, 1.2, 2\}$, and for DMW-AB $b=2$. The simulation is run for $2 \times 10^5$ slots, which is sufficiently long for convergence of the queue sizes. We first evaluate the performance of the DMW-RS algorithm with
Improved $\tau$-update and DMW-AB in terms of the maximum supported traffic
load and the average queue size,  by comparing them with the centralized MW algorithm. For all algorithms, in Fig. 1(a) we plot
the mean total queue size summed over all the  users, as the overall arrival rate increases.
Clearly, as the overall arrival rate exceeds approximately 5 packets/slot then queue sizes suddenly increase, and the network becomes unstable for all three algorithms, which also means that the  achievable total throughput with these algorithms is equal to 5 packets/slot. In addition, the total average queue size is almost the same with all algorithms for every arrival rate.  This result implies that  DMW-AB and DMW-RS   achieve  almost the same performance in terms of the throughput and average delay  as  that of the centralized MW policy since DMW-AB and DMW-RS schedule the user with the highest weight at every slot with very high probability. Also, as $b$ increases, the probability goes to one and, consequently, DMW-AB and DMW-RS  achieve the same performance as that of the centralized Max-Weight policy. 

In Fig. 1(b), for different values of $\delta$ we depict $\EE[M(t)]$ with DMR-RS as the overall arrival rate increases when $N=20$. It can be seen that  for all values of $\delta$ and when the network is highly loaded and even unstable,  $\EE[M(t)]$ tends to be less than 5 mini-slots, and DMW-RS is robust to the parameter $\delta$.  
%{\color{red} \st{In Figure 1-b, to see the effect of $b$,  we
%depict the average probability of scheduling the user with the
%maximum weighted throughput denoted $P_k$ with DMW-RS algorithm  vs. $b$ when the total
%average traffic load is  4.4 packets/slot. To this end, we set $\Bc=\{b\}$, and increase $b$ at each simulation. Clearly, as $b$ increases
%$P_k$ increases to one logarithmically, which matches our theoretical results. }}
Lastly, in Figure 1(c),  we depict $\EE[M(t)]$ as the overall arrival rate increases with different number of users (i.e., $N$). When $N=50$ and the arrival rate is equal to $1.8$ , and $N=60$ with arrival rate of $1.6$, the network is not stable, and at these points  $\EE[M(t)]$ starts increasing.   However, $\EE[M(t)]$ is less than 4 mini-slots for all
simulation scenarios even if the channel and arrival processes are highly dynamic (i.e., independent over time.) and the number of users and the traffic rate are high.  
%We observe that for all the settings in Fig. 1(b) and Fig. 1(c), DMW-RS and the centralized MW policy have the same performance in throughput and average delay.
\begin{figure*}[t!]
\centering
\begin{floatrow}
\ffigbox[\FBwidth]
{
\subfloat[Avg. total queue sizes vs. overall arrival rate ]{\includegraphics[height=3.5cm,width=5.5cm]{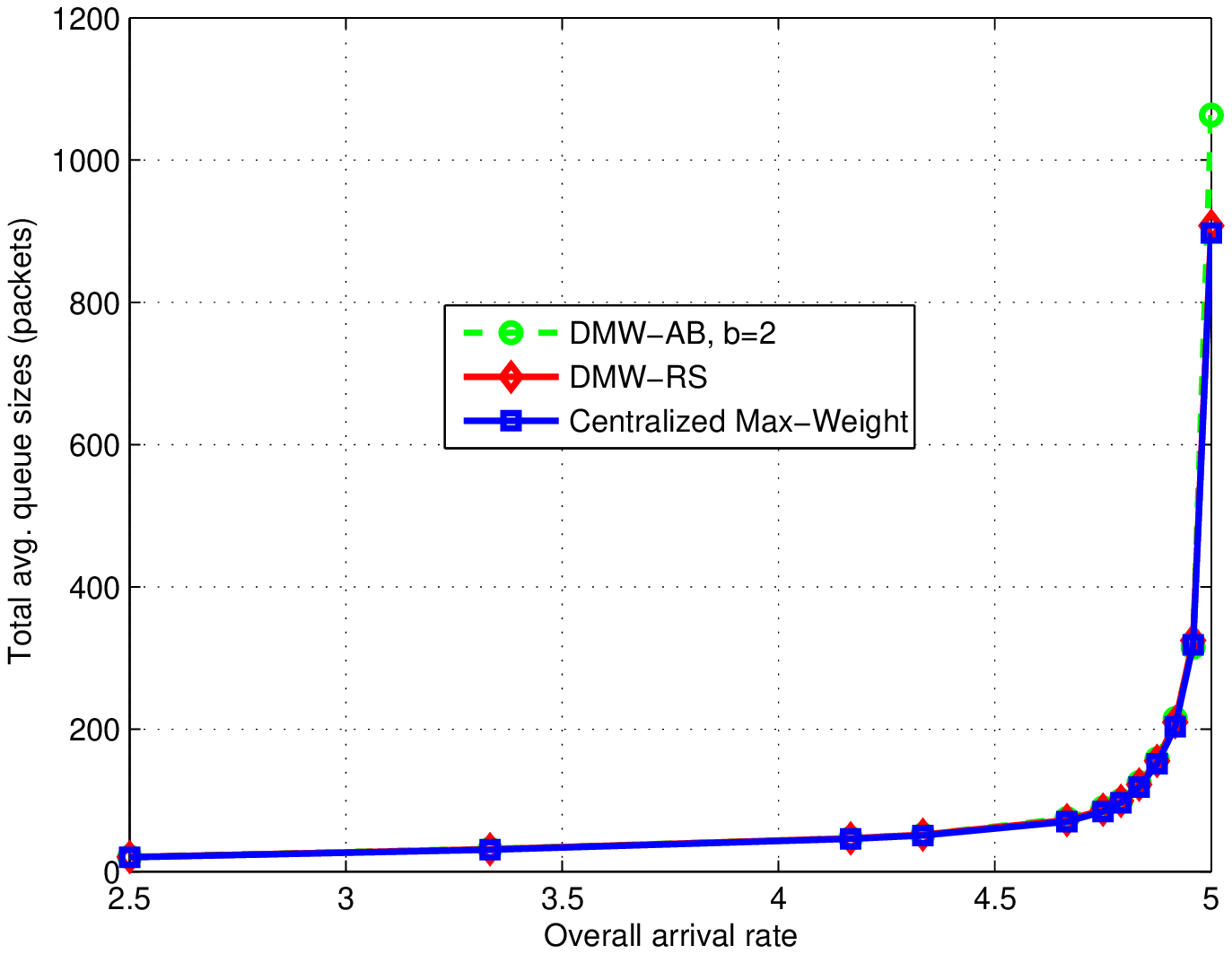}}
\quad
\subfloat[Avg. number of mini-slots with different $\delta$]{\includegraphics[height=3.5cm,width=5.5cm]{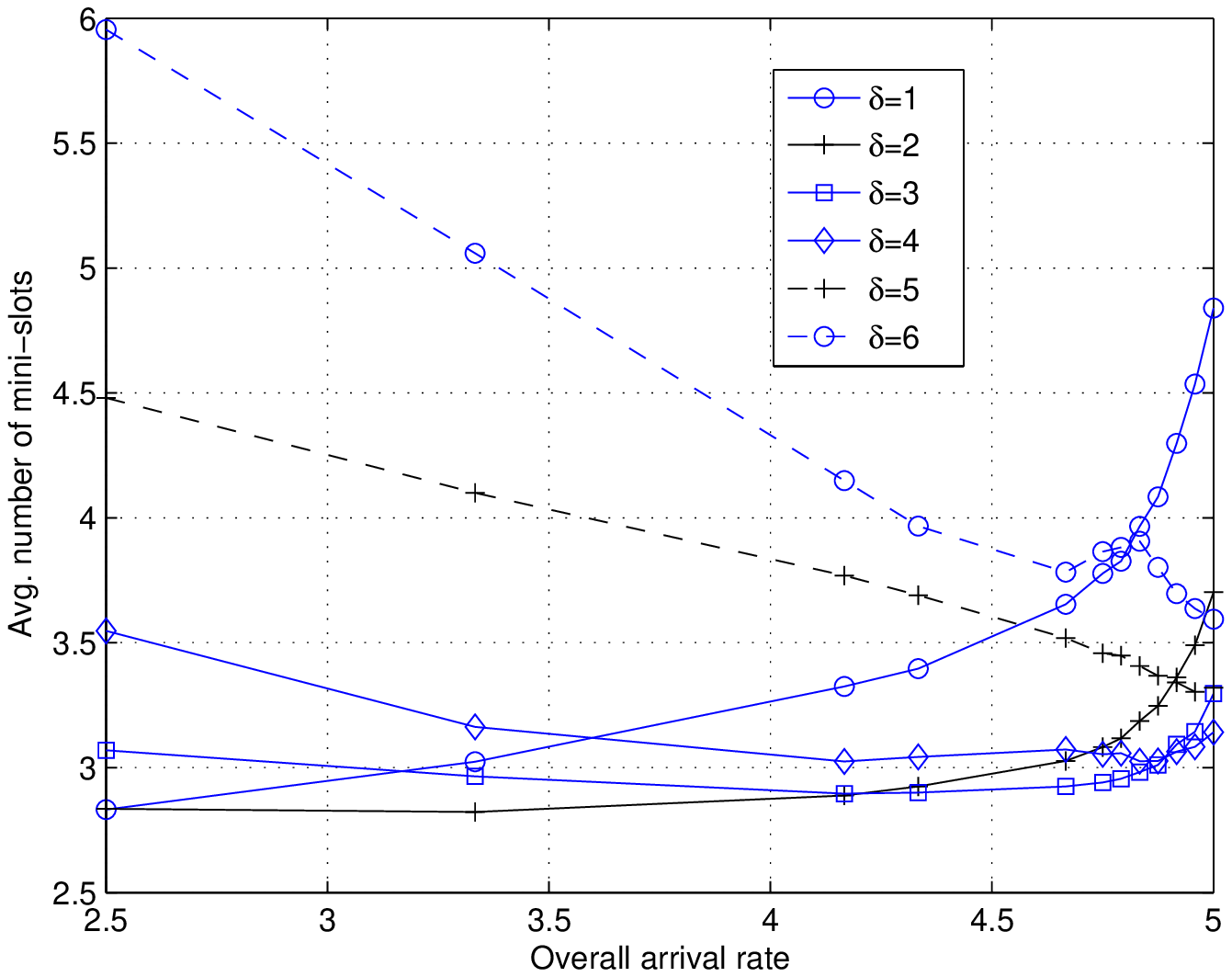}}
\quad
\subfloat[ Avg. number of mini-slots with different $N$]{\includegraphics[height=3.5cm,width=5.5cm]{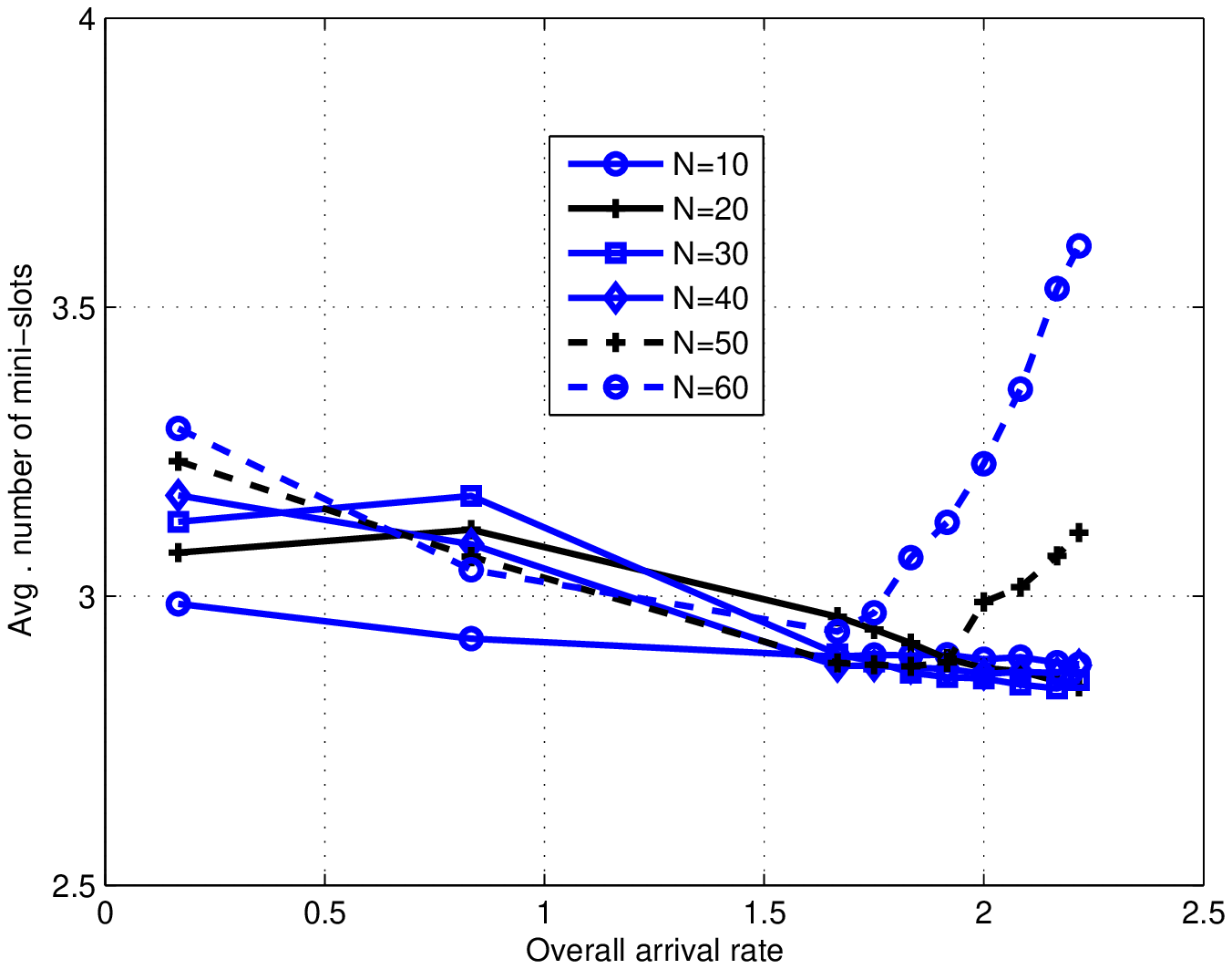}}
}{\caption{Performance of DMW-RS}}
%\ffigbox{
%\subfloat[Caption]{\includegraphics[height=5cm,width=5cm]{dmw_rs_EM.eps}}
%}{\caption{Caption here}}
\end{floatrow}
\end{figure*}
\vspace{-0.3cm}
\section{Conclusion }
\label{sec:conclusion}  In this letter, we have proposed two new
distributed scheduling policies for a fully-connected wireless
network over fading channels. By allowing arbitrary backoff time we have proved that the
first algorithm  behaves exactly the same as the centralized Max-Weight policy.  For practical systems, we then have developed a distributed algorithm
that operates with a reservation scheme with a low number of
mini-slots used for contention resolution, which can also achieve the same performance as the centralized Max-Weight policy.  Our future directions include designing optimal
algorithms to address the throughput-fairness and throughput-finite buffer trade-off  and also investigating the adaptation of DMW-RS to practical systems such as 802.11 based networks.

%distributed algorithms guarantee that the same users will be scheduled in every time slot as the centralized Max-Weight policy. 
%In this paper,
%we assumed that the network operator can allocate as many as
%mini-slots, which is not always the case in practical wireless
%networks. We will relax this assumption in our future work. Another
%future direction is to investigate scheduling algorithms for more
%general network topologies.
\vspace{-0.3cm}
\bibliographystyle{IEEEtran}
\bibliography{IEEEabrv,ref}
\vspace{-0.5cm}
\appendices
\section{Proof of Theorem \ref{thm:DMWAB}}
\label{sec:DMWAB} We prove Theorem \ref{thm:DMWAB} using Theorem
~\cite{Eryilmaz:ToN05}.  Let $w_{max}(t)=\max w_n(t)$ at time slot $t$.  According to the Theorem in [11], if given any $0 < \epsilon, \delta
<1$, there exists a constant $B >0$ such that: in any time  with probability greater than $1- \delta$, the complement of the following event $\chi$  occurs, where $\chi= \{l : \quad w_l(t) <
(1-\epsilon) w_{max}(t)\}$.
%\begin{align*}
%\chi= \{l : \quad w_l(t) <
%(1-\epsilon) w_{max}(t)\}
%\end{align*}
In order to show that $\pi(\chi') \geq 1-\delta$, which implies
Theorem 1, we next show that $\pi(\chi)=\sum_{l \in \chi} \pi_l < \delta$. Hence,
\begin{align*}
\pi(\chi)=\sum_{l \in \chi}
\frac{b^{w_k(t)}}{\sum_{i=1}^N  b^{w_i(t)}} \leq \frac{Nb^{(1-\epsilon)w^*(t)}}{\sum_{i= 1}^N  b^{w_i(t)}}
\end{align*}
Note that $\sum_{i =1}^N b^{w_i(t)} \geq b^{w_{max}(t)}$.
%\begin{align*}
%\sum_{i =1}^N b^{w_i(t)} \geq b^{w_{max}(t)}
%\end{align*}
Therefore, we have
\begin{align}
\pi(\chi) \leq \frac{Nb^{(1-\epsilon)w_{max}(t)}}{\sum_{i=1}^N
b^{w_i(t)}} \leq \frac{Nb^{(1-\epsilon)w_{max}(t)}}{b^{w_{max}(t)}}
\end{align}
By using (4), one can show that if the following inequality in (5) holds then $\pi(\chi) < \delta$:
\begin{align}
w_{max}(t) > \frac{\log_b(N) +
\log_b(\frac{1}{\delta})}{(\epsilon)}\label{eq:maxW}
\end{align}
Since $w_{max}(t)$ is a continuous and
nondecreasing function of queue sizes, and $R_n(t) \geq R_{min} >0$, with $\lim_{\textbf{Q}(t)\rightarrow \infty} w_{max}(t)=
\infty$, there exists $B>0$ such that
\begin{align}
w_{max}(t) > \max_n Q_n(t)R_{min} > \frac{\log_b(N) +
\log_b(\frac{1}{\delta})}{(\epsilon)}\label{eq:maxW}
\end{align}
Hence, if queue sizes are large enough we can find a constant $B$ such that
\begin{align}
\max_n Q_n(t) > \frac{\log_b(N) +
\log_b(\frac{1}{\delta})}{(R_{min}) \epsilon}\triangleq
B\label{eq:maxW}
\end{align}
Hence, Theorem 1 holds and $\pi(\chi) < \delta$. Thus, DMW-RS is
throughput-optimal.
\end{document}